\documentclass[letterpaper, 10 pt, twocolumn, conference]{ieeeconf}  % Comment this line out
\IEEEoverridecommandlockouts                              % This command is only
\usepackage{amsmath}
\usepackage{amssymb}
\usepackage{graphicx}
\usepackage{latexsym,epsfig,subfigure,dsfont}

\newtheorem{definition}{Definition}[section]
\newtheorem{remark}{Remark}[section]
\newtheorem{theorem}{Theorem}[section]

\newtheorem{lemma}{Lemma}[section]
\newtheorem{proposition}{Proposition}[section]

\newcommand{\eqdef} {\mbox{$\:\stackrel{\triangle}{=}\:$}}

\newcommand{\cI}{{\mathcal I}}

\newcommand{\cN}{{\mathcal N}}
\newcommand{\cM}{{\mathcal M}}

\newcommand{\cS}{{\mathcal S}}
\newcommand{\cT}{{\mathcal T}}
\def\b0{\mathbf{0}}

\def\bT{\mathbf{T}}

\def\bA{\mathbf{A}}
\def\bB{\mathbf{B}}
\def\bU{\mathbf{U}}
\def\bD{\mathbf{D}}
\def\bE{\mathbf{E}}
\def\bF{\mathbf{F}}

\def\bQ{\mathbf{Q}}
\def\bV{\mathbf{V}}
\def\bI{\mathbf{I}}
\def\bH{\mathbf{H}}
\def\bW{\mathbf{W}}
\def\bP{\mathbf{P}}
\def\bM{\mathbf{M}}

\def\bbR{\mathbb{R}}
\def\bx{\mathbf{x}}
\def\by{\mathbf{y}}
\def\bGam{\mathbf{\Gamma}}
\def\bgam{\mathbf{\gamma}}

\def\bz{\mathbf{z}}

\def\ba{\mathbf{a}}

\def\bv{\mathbf{v}}
\def\bs{\mathbf{s}}
\def\bbe{\mathbf{e}}
\def\bw{\mathbf{w}}

\def\bL{\mathbf{\Lambda}}
\def\bS{\mathbf{\Sigma}}
\def\bm{\mathbf{\mu}} 

\def\trQ{\textrm{Q}} 
\def\bzero{\mathbf{0}}

\newcommand{\argmax}{\operatornamewithlimits{argmax}}
\newcommand{\be}{\begin{equation}}
\newcommand{\ee}{\end{equation}}
\newcommand{\bea}{\begin{eqnarray}}
\newcommand{\eea}{\end{eqnarray}}
\newcommand{\bean}{\begin{eqnarray*}}
\newcommand{\eean}{\end{eqnarray*}}
\newcommand{\ben}{\begin{enumerate}}
\newcommand{\een}{\end{enumerate}}

\newcommand{\qed}{\hspace*{\fill}%
    \vbox{\hrule\hbox{\vrule\squarebox{.667em}\vrule}\hrule}\smallskip}
    \def\squarebox#1{\hbox to #1{\hfill\vbox to #1{\vfill}}}

\renewcommand{\theequation}{\arabic{section}.\arabic{equation}}
\newcommand{\Section}[1]{\section{#1}
\setcounter{equation}{0}}

\title{On Detection With Partial Information In The Gaussian Setup}

\author{Onur \"{O}zye\c{s}il,  M. K{\i}van\c{c} M{\i}h\c{c}ak, Y\"{u}cel Altu\u{g}
\thanks{
O. \"{O}zye\c{s}il is with PACM (the Program in Applied and Computational Mathematics), 
Princeton University, Princeton, NJ, 08544, 
{\tt\footnotesize{oozyesil@princeton.edu}}; 
M.~K.~Mihcak is with the Electrical and Electronics Engineering 
Department of Bo\u{g}azi\c{c}i University, Istanbul, 34342, Turkey,
{\tt\footnotesize{kivanc.mihcak@boun.edu.tr}};
Y. Altu\u{g} is with the School of Electrical and Computer Engineering, Cornell University, Ithaca, NY, 14853, 
{\tt\footnotesize{ya68@cornell.edu}} }
\thanks{M.~K.~M{\i}h\c{c}ak is partially supported by
T\"{U}B\.{I}TAK Career Award no. 106E117 and T\"{U}BA-GEBIP Award; 
during the time of research for this project, O.~\"{O}zye\c{s}il 
was partially supported by T\"{U}B\.{I}TAK Graduate Studies Fellowship, no. 2228, and
Y.~Altu\u{g} was partially supported by T\"{U}B\.{I}TAK Career
Award no. 106E117.}
}

\begin{document}

\maketitle
\thispagestyle{empty}
\pagestyle{empty}

%%%%%%%%%%%%%%%%%%%%%%%%%%%%%%%%%%%%%%%%%%%%%%%%%%%%%%%%%%%%%%%%%%%%%%%%%%%%%%%%
\begin{abstract}
We introduce the problem of communication with partial information, where there is an asymmetry 
between the transmitter and the receiver codebooks. 
Practical applications of the proposed setup include the robust signal hashing problem 
within the context of multimedia security and asymmetric  communications with resource-lacking receivers. 
We study this setup in a binary detection theoretic 
context for the additive colored Gaussian noise channel.
In our proposed setup, the partial information available at the detector consists of 
dimensionality-reduced versions of the
transmitter codewords, where the dimensionality reduction is achieved via a linear transform. 
We first derive the corresponding MAP-optimal 
detection rule and the corresponding conditional probability of error (conditioned on the partial information
the detector possesses). 
Then, we constructively quantify an optimal class of linear transforms, where the 
cost function is the expected Chernoff bound on the conditional probability of error of the MAP-optimal 
detector. 
\end{abstract}

\Section{Introduction}
\label{sec:intro}

In this paper, we introduce a communication-theoretic paradigm, which 
we name as ``communication with partial information'', and subsequently study 
it within a detection-theoretic context (therefore the term ``detection with partial information'') 
in a particular case of the Gaussian setup.
In the proposed paradigm, there is an inherent asymmetry between the information the 
transmitter and the receiver possess in terms of the utilized codebooks.
In particular, in the ``detection with partial information'' setup, 
the codebook of the receiver is formed via applying a non-invertible process
on the codebook of the transmitter; hence {\em the codebooks are different}. 
Thus, the information available at the transmitter forms a ``superset''
of the information available at the receiver. 
Note that, a reminiscent asymmetric structure between the 
transmitter and the receiver also exists in the well-known family of problems, 
termed as  ``communication with side information'' \cite{Sha58, Wyn75, Wyn76, Cos83}. 
However, in the paradigm of ``communication with side information'' (unlike the proposed 
``communication with partial information'' setup), the utilized codebooks
at the receiver and the transmitter are the same; in addition, either the transmitter
or the receiver is ``favored'' with the presence of ``extra'' information (which amounts to 
the ``side information'').

It appears that, there are at least two significant applications that motivate the formulation
of the ``communication with partial information'' approach:
\begin{itemize}
\item The first application 
can be viewed to fall within the category of ``robust signal hashing'' in the signal 
processing \& multimedia security literature \cite{Ven2000, Mih2001, Koz2002, Mon07}. 
In robust signal hashing,   a content owner provides
``robust  hash value''s of the protected content (that is some dimensionality-reduced versions 
of the protected content) to a third party, which
searches the content using its robust hash values as {\em the partial information} at the receiver end. 
These robust   hash values represent ``the content's significant features'' and are
ideally approximately-invariant under acceptable modifications to
the content. In practical applications, the third party that performs the hash-based search is
usually {\em not trusted}; hence, there is a significant issue of privacy. 
In particular, given a robust hash value, it should ideally be impossible to retrieve the 
original protected content from a privacy viewpoint. 
The setup proposed in this paper can be used as a detection-theoretic model to analyze the 
hash-based detection problem: 
the protected content is represented by the transmitted signal;
the robust hash values used in the search are represented by the partial information available
at the receiver; a perceptually-acceptable modification to the protected content 
is represented by the channel noise. 
\item The second application includes all instances of point-to-point communications, 
where there is an inherent asymmetry between the transmitter and the receiver
in terms of their storage capabilities and computational resources. In particular, the cases,
when the receiver is unable to store the codebook used by the encoder (due to a limit on the 
memory) or utilize the codebook used by the encoder (due to a limit on the computational 
resources), can be studied within the framework of ``communication with partial information''. 
In such cases, one potential remedy is the receiver's using a ``simplified'' (i.e., dimensionality-reduced) 
version of the codebook of the encoder. In practice, such situations may typically arise, for instance,  
when there is a bi-directional communication between a sensor and the base station (the resource-limited
receiver representing the sensor)  or when there is a bi-directional communication between 
a controller and a remote measurement unit. In such applications, the simplified version of the encoder
codebook is represented by the partial information at the receiver side. 
\end{itemize}

Our contributions in this paper can be listed as follows: 
\begin{itemize}
\item We introduce the paradigm of ``communication with partial information'' and study it within the context of 
binary detection in the Gaussian setup. We believe the main philosophy behind this formulation 
(i.e., introducing an asymmetry between the transmitter and the receiver in the sense of utilized codebooks) 
can be used to analyze various problems of interest in communication theory and signal processing. 
\item Within the binary hypothesis testing setup, we study a case, where the disturbance on the transmitter output consists of additive colored Gaussian noise, and the detector partial information is produced
via applying a linear (dimensionality-reducing) transform on the encoder codebook. 
Consequently, we present the following results: 
\begin{itemize}
\item We derive the MAP-optimal detection rule and the corresponding probability of error, both of which are conditioned
on the partial information available at the detector. 
\item We construct a class of {\em optimal} linear transforms, which minimize the expected 
(with respect to the joint distribution of the detector partial information) 
Chernoff bound on the aforementioned probability of detection error. 
\end{itemize}
\end{itemize}

In Sec.~\ref{sec:notation-problem}, we present the notation that is used throughout the paper and specify the formal
problem statement. In Sec.~\ref{sec:detection-partial}, we derive the MAP-optimal detection rule conditioned 
on the partial information available at the receiver. In Sec.~\ref{sec:opt-linear-expectation}, we quantify
an optimal (in the sense of the expected value of the Chernoff bound on the detection error probability) 
class of linear transforms that are used to generate the receiver partial information.  
We present illustrative numerical results in Sec.~\ref{sec:numerical-results}, followed by 
discussions and conclusion in Sec.~\ref{sec:conclusions}.

\Section{Notation and Problem Statement}
\label{sec:notation-problem}

\subsection{Notation}
\label{ssec:notation}

Boldface lowercase and uppercase letters denote vectors and matrices, respectively; the 
corresponding regular letters with subscripts denote their individual elements. For instance, 
given a vector $\ba$, $a_i$ represents its $i$-th element; given a matrix $\bA$, $A_{ij}$
denotes its $\left( i , j \right)$-th element. Note that, we do not use a separate notation for random 
vectors; we assume that it is clear from the context. 

Given a matrix $\bA$, $\bA^T$, $r \left( \bA \right)$ and $\det(\bA)$ denote its
transpose, rank and determinant, respectively; further, $\bI_n$ denotes the
identity matrix of size $n \times n$. Given the vectors $\bx , \by
\in \bbR^m$, $\langle \bx,\by \rangle$ indicates the inner product
that induces the Euclidean norm, i.e., $\langle \bx,\by \rangle = \sum_i x_i y_i$; 
accordingly the induced Euclidean norm is denoted by $\| \bx \| = \langle \bx , \bx \rangle^{1/2}$.

\begin{definition}
Given $\bA \in \bbR^{m \times n}$, such that $r\left(\bA\right)=k \leq \min \left( m , n \right)$,
\emph{Singular Value Decomposition} (SVD) of $\bA$ is unique (up to ordering) and defined as
\begin{equation}
\bA \eqdef \bU \bL \bV^T,
\label{eq:eq-SVD}
\end{equation}
where $\bU \in \bbR^{m \times k}$, $\bV \in \bbR^{n \times k}$, $\bL
\in \bbR^{k \times k}$ are called the left-singular vector matrix (orthonormal),
the right-singular vector matrix (orthonormal) and the singular value matrix of
$\bA$, respectively. The matrix $\bL$ is positive-definite diagonal;
we denote its entries along the diagonal by $\left\{ \sigma_i \left( \bA \right) \right\}_{i=1}^k$, 
which are the non-zero singular values of $\bA$, and assumed to be in non-increasing order
without loss of generality.  
% and these diagonal
%elements are termed as the singular values of $\bA$.
\label{def:def-SVD}
\end{definition}

%Given the matrix $\bA \in \bbR^{m \times k}$ of rank $r \leq \min \left( m , k
%\right)$, $\left\{ \sigma_i \left( \bA \right) \right\}_{i=1}^r$ denote its non-zero
%singular values, which are assumed to be in non-increasing order without loss of generality. 
%%%%%%%%%%%%%%%%%%%%%%%%%%%%%%%%%%%%%%%%%%%%%%%%
For a square matrix $\bA$ of size $k \times k$ and of rank $r \leq k$, 
$\left\{ \lambda_i  \left( \bA \right) \right\}_{i=1}^r$  denote its non-zero eigenvalues; 
in case $\bA$ is a symmetric matrix, $\left\{ \lambda_i \right\}$ are assumed to be 
in non-decreasing order. 
We use $\cN \left( \bm, \bS \right)$ to denote a
multivariate Gaussian distribution, with mean vector $\bm$ and
covariance matrix $\bS$. Furthermore, $\trQ \left( \cdot \right)$
denotes the standard $Q$-function:  $\trQ \left(
\alpha \right) \eqdef \int_{\alpha}^{\infty} \frac{1}{\sqrt{2 \pi}}
e^{-x^2/2} dx$.

\subsection{Problem Statement}
\label{ssec:problem}

We analyze a binary communication system, where the encoder selects one of the two codewords, 
$\bx_0$ and $\bx_1$, representing the message bit $i \in \left\{ 0 , 1 \right\}$, where $\Pr \left( i = 0 \right) = 
\Pr \left( i = 1 \right) = 1/2$; the selected codeword, $\bx = \bx_i$, is sent through a channel. 
The encoder output $\bx$ is  corrupted by an additive, signal-independent, 
(not necessarily white) Gaussian noise, denoted by $\bbe$, thereby yielding the overall channel output $\by$. Observing $\by$,
the receiver acts as a detector and makes a binary decision, as to the origins of received signal.
We pursue a detection-theoretic approach to solve this problem and assume uniform costs.  
We assume that $\bx_0$, $\bx_1$,  $\bbe$, and $\by$ are all length-$n$ real-valued vectors, where
$\bx_0$ and $\bx_1$ are independent of each other and 
$\bx_0, \bx_1 \sim \cN \left( \bzero , \bS_x \right)$, $\bbe \sim \cN \left( \bzero , \bS_e \right)$ is independent
of both $\bx_0$ and $\bx_1$. Here, we also assume that the covariance matrix of the original signals $\bS_x$ and the covariance matrix of the noise $\bS_e$ are {\em positive definite} (they are also symmetric by construction).
See Fig.~\ref{fig:system-diag} for a schematic illustration of the proposed problem. 

\begin{figure}[!htb]
\centering \centerline{\epsfxsize = 3.5in
\epsfbox{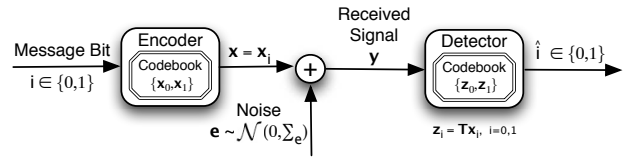}} \caption{Block diagram representation
of the problem of ``binary detection with partial information''.} 
\label{fig:system-diag}
\end{figure}

In the considered setup, {\em the detector 
does not know the original codewords $\left\{ \bx_0 , \bx_1 \right\}$, but only their distributions and
their dimensionality-reduced versions, $\left\{ \bz_0 , \bz_1 \right\}$}, where $\bz_i = \bT \cdot \bx_i$, $i = 0 , 1$, and $\bT$
is a deterministic real matrix of size $m \times n$, $m < n$, $r \left( \bT \right) = m$. 
Note that, this implies, $\bz_0$ and $\bz_1$  are both length-$m$ real-valued vectors. 
As such, the proposed problem is radically different from the conventional binary detection scenario due to the 
{\em mismatch between the codebooks of the encoder and the detector}. 
Consequently, we term the problem at hand as ``detection with partial information'' for the Gaussian case.

An important point here is that, since the receiver fully knows the statistical characterization of the whole
system, it is able to apply the MAP decoding rule. 
In particular, in Sec.~\ref{sec:detection-partial}, we derive the MAP detection rule, which is given as a function
of the partial information $\left( \bz_0 , \bz_1 \right)$, and the corresponding conditional probability of error
(conditioned on $\bz_0$ and $\bz_1$). 
Subsequently, in Sec.~\ref{sec:opt-linear-expectation},
we derive the optimal linear transform, $\bT$, in the sense of 
the expected Chernoff bound on the conditional probability of error 
of the MAP detector.

\begin{remark}
In \cite{Mih08}, the authors study a closely-related problem, which can be viewed as the 
``deterministic variant'' of the aforementioned setup. In particular, in \cite{Mih08} 
the authors assume that the encoder codewords $\left\{ \bx_i \right\}$ are deterministic, unknown
and the subsequent analysis is based on the probability of error induced by the GLRT (generalized
likelihood ratio test) rule. On the other hand, in this paper, we assume that the encoder
codewords $\left\{ \bx_i \right\}$ are random (in particular Gaussian) and perform a MAP-based analysis. 
\end{remark}

\begin{remark}
Although the problem imposed in this paper is the binary detection case, the analysis 
can be extended to apply a ``union bound based approach'' for the $L$-ary case with little or no 
difficulty\footnote{In the $L$-ary case, the message is $\log L$ bits long; the encoder and receiver codebooks are
$\left\{ \bx_i \right\}_{i=0}^{L-1}$ and $\left\{ \bz_i \right\}_{i=0}^{L-1}$, respectively.}. 
A similar approach and discussion was provided in \cite{Mih08} for the case of deterministic $\left\{ \bx_i \right\}$. 
\end{remark}

\Section{Optimal Detection Conditioned On The Partial Information}
\label{sec:detection-partial}

At the detector side, we are given $\left\{\bz_0,\bz_1\right\}$, which yield partial information 
about the true codewords $\left\{ \bx_0 , \bx_1 \right\}$.  The {\em binary hypothesis testing} 
approach on the detector side utilizes the MAP detection rule \cite{Poor88}: It operates on 
the observed data $\by$ (generated by the process explained in Sec.~\ref{ssec:problem}), 
and makes a binary decision regarding the message bit given $\left\{ \bz_0 , \bz_1 \right\}$. 
%Since we have equal priors and uniform costs on the hypotheses, 
%the  MAP detection rule simplifies to {\em the maximum likelihood detection rule} \cite{Poor88}. 
Thus, we aim to solve the following binary hypothesis testing problem:
\bean
H_0 &  : & \by = \bx_0 + \bbe\hspace{0.1in};\hspace{0.1in} \quad \mbox{given $\left\{\bz_0,\bz_1\right\}$}, \\
H_1 &  : & \by = \bx_1 + \bbe\hspace{0.1in};\hspace{0.1in} \quad \mbox{given $\left\{\bz_0,\bz_1\right\}$} .
\eean
The corresponding MAP detection rule is given by 
\be
p \left( \by | H_0 \right) \mathop{\gtrless}_{H_1}^{H_0} 
p \left( \by | H_1 \right).
\label{eq:MLdetection-rule}
\ee
since we have equal priors and uniform costs. Note that, (\ref{eq:MLdetection-rule}) is also known as the maximum-likelihood
detection rule \cite{Poor88}. 
Note that, for all $i \in \left\{ 0 , 1 \right\}$, we have 
\[
p \left( \by \, | \, H_i \right)  
%=  p \left( \by \, \big| \, \mbox{$\bx_i$ is sent} , \, \left\{ \bz_0 , \bz_1 \right\} \right)
= p \left( \bx_i + \bbe \, \big| \, \bz_i \right) \Big|_{\bx_i + \bbe = \by}, 
\]
which implies that (\ref{eq:MLdetection-rule}) 
can be rewritten as  
\be
\left. p \left( \bx_0 + \bbe  | \bz_0 \right) \right|_{\bx_0 + \bbe = \by} \mathop{\gtrless}_{H_1}^{H_0} 
\left. p \left( \bx_1 + \bbe | \bz_1 \right) \right|_{\bx_1 + \bbe = \by}.
\label{eq:MLdetection-rule2}
\ee

\begin{theorem}
The maximum likelihood detection rule~(\ref{eq:MLdetection-rule2}) is given by 
\be
\| \bS_{ y | z }^{-1/2} \left( \by - \bm_{ y_0 | z_0 } \right) \|  \mathop{\gtrless}_{H_0}^{H_1}
\| \bS_{ y | z }^{-1/2} \left( \by - \bm_{ y_1 | z_1 } \right) \| 
\label{eq:ML_detect_rule}
\ee
The corresponding (conditional) probability of error (conditioned on $\bz_0$ and $\bz_1$) is given by
\be
%\Pr \left[ \mbox{error} \, | \, \left\{ \bz_0, \bz_1 \right\} \right] = 
P_{e | \bz_0 , \bz_1} = 
\trQ \left( \frac{ \|   \bS_{ y | z }^{-1/2} \left( \bm_{ y_0 | z_0 }  - \bm_{ y_1 | z_1 } \right) \| }{2} \right) 
\label{eq:cond_prob_err}
\ee
where, for $i \in \left\{ 0 , 1 \right\}$,
$\bm_{ y_i | z_i }   =  
\left. \textrm{E} \left( \by_i \, | \, \bz_i \right) \right|_{ \by_i = \bx_i +  \bbe} = \bS_x \bT^T \left( \bT \bS_x \bT^T \right)^{-1} \bz_i$;
$\bS_{y|z}$ is positive definite and given by 
%\begin{equation}
$
\bS_{y|z}  =  \left. \textrm{Cov} \left( \by_i \, | \, \bz_i \right) \right|_{ \by_i = \bx_i +  \bbe, \, i = 0 , 1} = 
\bS_x + \bS_e - \bS_x \bT^T \left( \bT \bS_x \bT^T \right)^{-1} \bT \bS_x  . 
$
%\label{eq:sigmay|x}
%\end{equation}
%and $\bS_{y|z}$ is positive definite.
\label{thm:cond_pro_err}
\end{theorem}
\begin{proof}
See Appendix~\ref{prf:cond_pro_err}.
\end{proof}

\begin{remark}
Using Theorem~\ref{thm:cond_pro_err}, we see that, if $\bz_0 = \bz_1$, conditional probability of error is $1/2$, 
which is meaningful. Then, 
there is nothing to discriminate from the detector's perspective thereby converting the detection to a fair coin toss. 
\end{remark}

\begin{remark}
The argument of the $Q$-function in (\ref{eq:cond_prob_err}) is always non-negative.  
This  allows us to set a tight bound on the expected probability of error, and analyze 
it in Sec.~\ref{sec:opt-linear-expectation}.
\end{remark}

\Section{Optimal Linear Operators In The Expectation Sense}
\label{sec:opt-linear-expectation}

In this section, our performance criterion is based on the {\em expected} (unconditional) probability of error of the MAP 
detector, denoted by $P_e$, given by 
\begin{eqnarray}
 \hspace{-0.6cm}P_e & = & 
%\textrm{E}_{\{\bz_0,\bz_1\}}\left(\Pr\left[\mbox{error} \, \big| \, \{\bz_0,\bz_1\}\right]\right) , 
\textrm{E}_{\{\bz_0,\bz_1\}} \left[ P_{e | \bz_0 , \bz_1 } \right] , 
\nonumber \\
 \hspace{-1cm} & \hspace{-0.1cm} = & \hspace{-0.1cm}
\textrm{E}_{\{\bz_0,\bz_1\}}\left(\trQ \left( \frac{ \|   \bS_{ y | z }^{-1/2} \left( \bm_{ y_0 | z_0 }  - \bm_{ y_1 | z_1 } \right) \| }{2} 
 \right)\right) 
\label{eq:expected_error}
\end{eqnarray}
%\end{equation}
where $\textrm{E}_{\{\bz_0,\bz_1\}}\left(.\right)$ denotes expectation with respect to the joint distribution of $\bz_0$ 
and $\bz_1$, and the right hand side follows from (\ref{eq:cond_prob_err}). 

\begin{remark}
It appears to be manageable to find a linear transform that minimizes the conditional probability of error, 
%$\Pr\left[\mbox{error} \, \big| \, \{\bz_0,\bz_1\}\right]$ 
$P_{e | \bz_0 , \bz_1}$
(see, for instance, \cite{Mih08}) as a function of the transmitted 
signals, $\bx_0$ and $\bx_1$, which would yield an ``input-adaptive optimal transform''. 
On the other hand, the expected probability of error given by (\ref{eq:expected_error}) is not tractable for an analogous analysis,
carried  out to characterize the optimal linear transform $\bT$ that minimizes it. 
This stems from the fact that, such an optimal $\bT$ would be a function of the overall statistics of the system
(corresponding to applying the operator of  $\textrm{E}_{\{\bz_0,\bz_1\}}\left(.\right)$ in (\ref{eq:expected_error}))
rather than individual realizations, which yields a ``complicated'' cost function to minimize; 
the result of the expectation operation, i.e., the $m\times m$-fold integration in~(\ref{eq:expected_error}) is 
not given in terms of standard analytical functions.
Therefore, we continue our analysis by characterizing linear operator(s) that minimize {\em a tight upper bound} on the expected 
probability of error defined by~(\ref{eq:expected_error})
\label{rem:expected-reasoning}
\end{remark}

Hence, we proceed with the following approach:  
We first bound 
%$\Pr\left[\mbox{error} \, \big| \, \{\bz_0,\bz_1\}\right]$ 
$P_{e | \bz_0 , \bz_1}$
for any given pair of $\left\{\bz_0  ,  \bz_1 \right\}$ 
from above and make use of the fact that expected value of this upper bound is an upper bound on $P_e$
(since, by definition, 
$P_{e | \bz_0 , \bz_1} \geq 0$). 
%$\Pr \left[ \mbox{error} \, \big| \, \left\{\bz_0  ,  \bz_1 \right\} \right] \geq 0$). 
Also, note that the use of an {\em upper} bound clearly makes sense since we aim
to {\em minimize}  $P_e$. The upper bound on 
$P_{e | \bz_0 , \bz_1}$
%$\Pr\left[\mbox{error} \, \big| \, \left\{\bz_0  ,  \bz_1 \right\} \right]$
that we use is the {\em Chernoff bound} on the $\trQ$-function (see {\em Basic Inequality} in~\cite{Loeve}), which is an 
exponentially decaying and a sufficiently tight bound. The expected Chernoff bound, which replaces the primary objective function $P_e$ in the design of optimal linear transfom $\bT$ due to its analytical tractability and sufficient tightness, is derived in the following proposition.

\begin{proposition}
The Chernoff bound on 
$P_{e | \bz_0 , \bz_1}$ is 
%$\Pr \left[\mbox{error} \, \big| \, \left\{\bz_0 ,  \bz_1 \right\} \right]$ is given by
\be
%\Pr \left[\mbox{error} \, \big| \, \left\{\bz_0  ,  \bz_1 \right\} \right] 
\hspace{-0.07cm}
P_{e | \bz_0 , \bz_1}
\leq \frac{1}{2} \exp 
\left(-\frac{ \| \bS_{ y | z }^{-1/2} \left( \bm_{ y_0 | z_0 }  - \bm_{ y_1 | z_1 } \right) \|^2 }{8}\right) , 
\label{eq:chernoff_bound}
\ee
yielding the following corresponding ``expected Chernoff bound'' on $P_e$ 
\be
P_e \leq \frac{1}{2}\left\{\det \left( \bI_m+ \frac{1}{2} \bW \right) \right\}^{-1/2}
%\frac{\bT \bS_x \bS_{y \, | \, z}^{-1} \bS_x \bT^T \left( \bT \bS_x \bT^T \right)^{-1}}{ 2 } \right) \right\}^{-\frac{1}{2}}
\label{eq:Exp_Chernoff_bound}
\ee
where $\bW \eqdef \bT \bS_x \bS_{y \, | \, z}^{-1} \bS_x \bT^T \left( \bT \bS_x \bT^T \right)^{-1}$ 
\label{prop:Exp_Chernoff_bound}
\end{proposition}
\begin{proof}
See Appendix~\ref{prf:Exp_Chernoff_bound}.
\end{proof}

\begin{remark}
The bound on expected (unconditional) probability of error of the MAP detector, given by~(\ref{eq:Exp_Chernoff_bound}) 
is the objective function we aim to minimize in this section. The minimization (over $\bT$) 
is carried out over a class of linear transformations that posses certain properties
imposed by the physical structure of the analyzed system. 
The obvious one of these properties is the dimension of the transformation (i.e., the fact that $\bT$ is a $m \times n$ matrix); 
the other one is the constraint on its rank (i.e., the fact that $r \left( \bT \right) = m$). 
The rank constraint is set to ensure that the dimensionality of the subspace (which is equal to $r \left( \bT \right)$), 
to which the partial information shared by the two sides of the communication belongs, is at a certain desired level;
this is because of the following fact: the performance of a system, which utilizes a rank-deficient transformation, 
is analogous to the performance of another system, the transformation of which is full-rank and has the same rank 
as the previous rank-deficient transformation.
\label{rem:optimalT-rank}
\end{remark}
 
\begin{definition}
\label{def:opt_trans}
The ``expected probability of error bound minimizing transform $\bT_{opt}$'' is given by 
\be
\bT_{opt} =  \argmax_{\substack{\bT\in\bbR^{m\times n}\\  r\left( \bT \right) = m}} \quad  
\det \left( \bI_m + \frac{1}{2} \bW \right)
%\frac{\bT \bS_x \bS_{y \, | \, z}^{-1} \bS_x \bT^T \left( \bT \bS_x \bT^T \right)^{-1}}{2}\right)
\label{eq:optimal-T-exp}
\ee
\end{definition}

\begin{proposition}
Let 
$\cS_\bT \eqdef \left\{ \bT \, \big| \, \bT \in \bbR^{m\times n}, r \left( \bT \right) = m \right\}$, 
$\cS_\bM \eqdef \left\{ \bM \, \big| \, \bM \in \bbR^{n\times m}, \bM^T \bM = \bI_m \right\}$, 
$\bP \eqdef \bL^{-1} \bF^T \left( \bS_x \right.$ $\left. + \bS_e \right) \bF \bL^{-1}$. 
Let the SVD of $\bS_x$ and $\bP$ be given by $\bS_x = \bF \bL^2 \bF^T$ and $\bP = \bU_p \bL_p \bU_p^T$, 
respectively, and $\hat{\bL}_p \eqdef \bI_n - \bL_p^{-1}$. Also define 
\[
G \left( \bM \right) \eqdef \left( \frac{1}{2} \right)^m \prod_{i=1}^m \left[ 1 + \frac{1}{\lambda_i \left( \bM^T \hat{\bL}_p \bM \right) }
\right]
\]
%\qquad 
\[
J \left( \bT \right) \eqdef \det \left[ \bI_m + \frac{1}{2} \bW \right]
%J \left( \bT \right) \eqdef \det \left[ \bI_m + \frac{1}{2} \bT \bS_x \bS_{y|z}^{-1} \bS_x \bT^T \left( \bT \bS_x \bT^T \right)^{-1} \right]
. 
\]
Suppose there exists 
\begin{equation}
\bM^\ast = \argmax_{\bM \in \cS_\bM} G \left( \bM \right). 
\label{eq:optimal-reduct-M}
\end{equation}
Then, letting
$\bT^\ast \eqdef \bE \bD \left( \bM^\ast \right)^T \bU_p^T \bL^{-1} \bF^T$, where $\bE \in \bbR^{m\times m}$ is an arbitrary 
unitary matrix and $\bD \in \bbR^{m \times m}$ is an arbitrary diagonal positive-definite matrix, we have 
$\bT^\ast = \argmax_{\bT \in \cS_\bT} J \left( \bT \right)$. 
\label{prop:optimal-T-prb-reduction}
\end{proposition}
\begin{proof}
See Appendix~\ref{prf:optimal-T-prb-reduction}.
\end{proof}

Proposition~\ref{prop:optimal-T-prb-reduction} allows us to deduce the existence of $\bT_{opt}$ with the sufficiency of the existence of $\bM^\ast$. Then, in order to find an optimal linear transformation, which is the main goal of this section, we first need to show the existence of $\bM^\ast$, and then construct $\bT_{opt}$ using $\bM^*$ that is the solution for the reduced problem~(\ref{eq:optimal-reduct-M}).

\begin{proposition}
A set of solutions for~(\ref{eq:optimal-reduct-M}) is given by
%\be
%\label{eq:opt-soln-reduced-prb}
\[
\cM = \left\{\bM \in \cS_{\bM} \; \Bigg| \; \bM = \bQ^T 
\left[\begin{array}{c}\bGam_m \\\b0_{\left(n-m\right)\times m}\end{array}\right] 
\right\}, 
%\bGam_m^T \bGam_m = \bI_m \right\} , 
\]
%\mbox{ $\bGam_m \in \bbR^{m\times m}$ unitary} \right\} , 
%is a unitary matrix}\right\} ,
%\ee
where $\bGam_m \in \bbR^{m\times m}$ is a unitary matrix,
$\bQ \in \left\{ 0 , 1 \right\}^{n \times n}$ denotes a permutation matrix s.t. the eigenvalues of 
$\bQ \hat{\bL}_p \bQ^T$ are in non-decreasing order. Moreover, 
\begin{equation}
\max_{\bM \in \cS_\bM} \prod_{i=1}^m \left[ 1 + \frac{1}{\lambda_i \left( \bM^T \hat{\bL}_p \bM \right)} \right]
= \prod_{i \in \cI} \left[ 1 + \frac{1}{\lambda_i \left( \hat{\bL}_p \right) } \right] , 
\label{eq:optimal_det_val}
\end{equation}
where $\cI \subseteq \left\{ 1 , 2 , \ldots , n \right\}$ denotes the cardinality-$m$ index set corresponding to the $m$-smallest 
eigenvalues of $\hat{\bL}_p$. 
%Moreover,
%\be
%\max_{\bM \in \cS_{\bM}} \quad \prod_{i=1}^m \quad 1+\frac{1}{\lambda_i\left( \bM^T \hat{\bL}_{\bP} \bM \right)}  \quad =  \quad
%\prod_{i=1}^m  \quad 1+\frac{1}{\lambda_i\left(\hat{\bL}_{\bP}\right)}.
%\label{eq:optimal_det_val}
%\ee
\label{prop:opt-reduced-soln}
\end{proposition}
\begin{proof}
See Appendix~\ref{prf:opt-reduced-soln}.
\end{proof}

\begin{theorem} 
\label{thm:optimal_lin_trans}
A set of optimal linear transforms, in the sense of expected Chernoff bound on the probability of error $P_e$, for communication with partial information in the Gaussian setup is given by
\be
\label{eq:opt_lin_trans_set}
\cT =  \left\{\bT \in \cS_{\bT} \; | \; \bT = \bE \bD \bM^T \bU_{p}^T \bL^{-1} \bF^T\right\}
\ee
%; 
%\mbox{ $\bE, \bD \in \bbR^{m\times m}$ are unitary and diagonal respectively and 
%$\bM \in \cM$}\right\},
%\ee
where 
$\bE \in \bbR^{m\times m}$ is unitary, 
 $\bD \in \bbR^{m\times m}$ is diagonal, 
$\bM \in \cM$, 
$\cS_{\bT} = \left\{ \bT \in \bbR^{m\times n} \, | \, r \left( \bT \right) = m \right\}$, $\cM$ is given by Proposition~\ref{prop:opt-reduced-soln} and $\bF$, $\bL$ and $\bU_{p}$ denote matrix of eigenvectors and diagonal matrix of eigenvalues of $\bS_x$ and the matrix of eigenvectors of $\bP = \bL^{-1}\bF^T\left( \bS_x + \bS_e \right) \bF \bL^{-1}$, respectively. 
%Moreover, the cardinality of $\cT$ is equal to that of the continuum.
\end{theorem}

\begin{proof}
By Proposition~\ref{prop:optimal-T-prb-reduction} we know that $\cT \neq \emptyset$.
We also know for a given $\bM^*$, i.e. $\bM$ satisfying~(\ref{eq:optimal-reduct-M}), $\bT = \bE \bD \bM^{*^T}\bU_{p}^T\bL^{-1}\bF^T$ satisfies~(\ref{eq:optimal-T-exp}), i.e., $\bT = \bT_{opt}$ (cf. Appendix~\ref{prf:optimal-T-prb-reduction}). 
Moreover, a set of $\bM$ satisfying~(\ref{eq:optimal-reduct-M}), namely $\cM$, is given by Proposition~\ref{prop:opt-reduced-soln}. This clearly implies that $\cT$, induced by $\cM$, is a set of optimal linear transforms, in the sense of expected Chernoff bound on the probability of error $P_e$.
%, for communication with partial information in the Gaussian setup. 
%Finally, since by Proposition~\ref{prop:optimal-T-prb-reduction} we know that $\cM \neq \emptyset$ implies that cardinality of $\cT$ is that of the continuum, we are done by Proposition~\ref{prop:opt-reduced-soln}.
\end{proof}

\Section{Numerical Results}
\label{sec:numerical-results}

\noindent
\underline{\em Optimality of $\bT^\ast$}:
Theorem~\ref{thm:optimal_lin_trans} gives a set of optimal linear transforms, however does not address the ``denseness'' of $\cT$ in $\cS_{\bT}$: ``is it easy to find an optimal transform in $\cS_{\bT}$ randomly, and how much is the performance of transforms in $\cS_{\bT} \backslash \cT$ separated from that of optimal transforms?''. The computational provided in Fig.~\ref{fig:optTvsrandomT} provide an experimental basis.
In  Fig.~\ref{fig:optTvsrandomT}, the simulations are performed with $\bS_x$ and $\bS_e$ having uniformly distributed eigenvalues, and the result is given using the reciprocal of the Chernoff bound on $P_e$ to improve visibility. 
%This figure shows two simple results. One result 
The first observation is that it is not ``easy'' to guess an element of $\cT$ randomly (we actually simulated over much larger number of trials, however give here the result for a set of $1000$ trials for illustrative purposes). This is clear by observing that none of the transforms chosen randomly from $\cS_{\bT}$ achieves the optimal value calculated from~(\ref{eq:optimal_det_val}) in Proposition~\ref{prop:opt-reduced-soln}, except $\bT_{opt}$ constructed by~(\ref{eq:opt_lin_trans_set}) and indicated as the transform in the middle of set of transforms, i.e. $\bT_{500}$. Also, the minimum value of the bound on $P_e$ achieved by arbitrary choices is not even close to that achieved by $\bT_{opt}$, it is around $4$ times larger than the minimum bound on $P_e$. Thus, we experimentally conjecture that $\cT$ is not ``dense'' in $\cS_{\bT}$. \\
% (moreover, the difference in performances is much larger for other choices of $n, m$ and $\textrm{E}\left( \|\bx\|^2 \right)/ \textrm{E}\left(\|\bbe\|^2 \right)$, but again we hold this difference in a visible range). By these two results, it is possible to claim that $\cT$ is not ``dense'' in $\cS_{\bT}$, i.e. construction for optimality gives a distinguished performance in the sense of Chernoff bound on $P_e$.
\begin{figure}[!htb]
\centering \centerline{\epsfxsize = 3in
\epsfbox{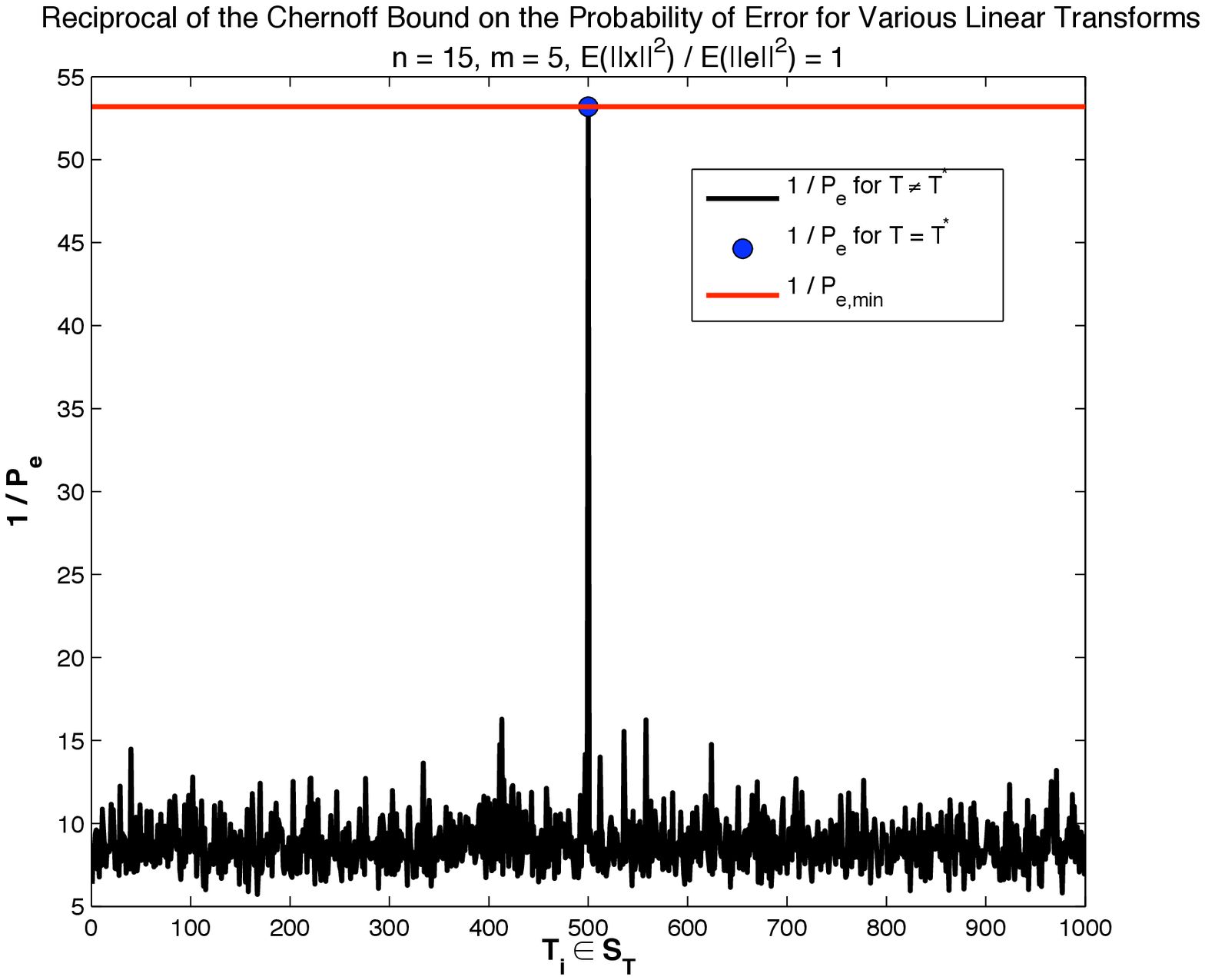}} \caption{Performance of $\bT_{opt}$ compared to arbitrary $\bT \in \cS_{\bT}$} 
\label{fig:optTvsrandomT}
\end{figure}

\noindent
\underline{\em $P_e$ vs. $\textrm{E}\left( \|\bx\|^2 \right)/ \textrm{E}\left(\|\bbe\|^2 \right)$}: In this part we observe the effect of  
$SNR = \textrm{E}\left( \|\bx\|^2 \right)/ \textrm{E}\left(\|\bbe\|^2 \right)$ on the optimality of $\bT_{opt}$. 
Fig.~\ref{fig:PevsSNR} is given to discuss this effect.
Similar to the setup of top-left panel,  the simulations are performed with $\bS_x$ and $\bS_e$ having uniformly distributed eigenvalues. As expected, the performance at optimality improves with increasing SNR since 
it gets easier to differentiate $\bz_0$ from $\bz_1$ in that case. \\
%The result is clear and expected; as $SNR$ increases the performance at optimality is improved, since for higher $SNR$ it is possible to discriminate $\bz_0$ from $\bz_1$ better, which then improves detection if we consider ML detection rule given by~(\ref{eq:ML_detect_rule}) and $\bm_{ y_i | z_i }  = \left. \textrm{E} \left( \by_i \, | \, \bz_i \right) \right|_{ \by_i = \bx_i +  \bbe} = \bS_x \bT^T \left( \bT \bS_x \bT^T \right)^{-1} \bz_i $.
\begin{figure}[!htb]
\centering \centerline{\epsfxsize = 3in
\epsfbox{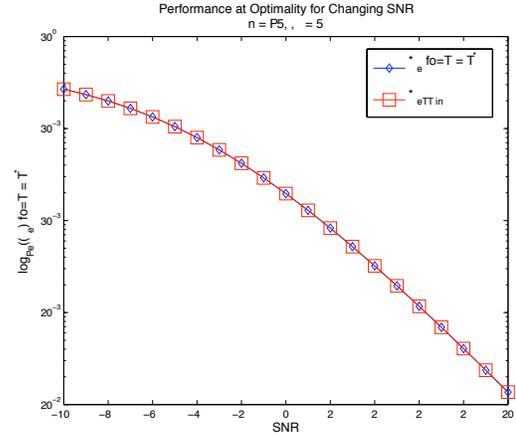}} \caption{Performance of $\bT_{opt}$ vs. SNR (dB), $P_e$ indicates Chernoff bound on expected probability of error here} 
\label{fig:PevsSNR}
\end{figure}

\noindent
\underline{\em $P_e$ vs. $m$}:
In this case, we study the effects of the amount of partial information shared by the detector side on the bound on the expected performance of the detector. 
%In other words, we observe the effect of the information available to the detector to construct its ``codebook'' for the binary data. 
This case is studied for $\bS_x$ and $\bS_e$ having uniformly distributed eigenvalues and $SNR = 1$. For $n = 50$, we construct $\bT_{opt}$ for particular values of $m$ and evaluate its performance in the sense of expected Chernoff bound on $P_e$. Results are shown  in Fig.~\ref{fig:PevsM}.
As expected, the capability of the detector improves as the amount of partial information increases. 
%Again the result is as expected; for increasing amount of partial information the capability of the detector to make a decision between $\bz_0$ and $\bz_1$ increases, so we observe a decreasing optimal bound, which is achieved by $\bT_{opt,e}$. 
Also, as $m$ tends to $n$, the performance at optimality converges to that for $m = n$, which is the Gaussian bound (the case 
when  $\bT$ is invertible).\\
% given by $\frac{1}{2}\left\{\det \left( \bI_n+\frac{\bS_x \bS_e^{-1}} { 2 } \right) \right\}^{-\frac{1}{2}}$, derived from~(\ref{eq:Exp_Chernoff_bound}). \\
\begin{figure}[!htb]
\centering \centerline{\epsfxsize = 3in
\epsfbox{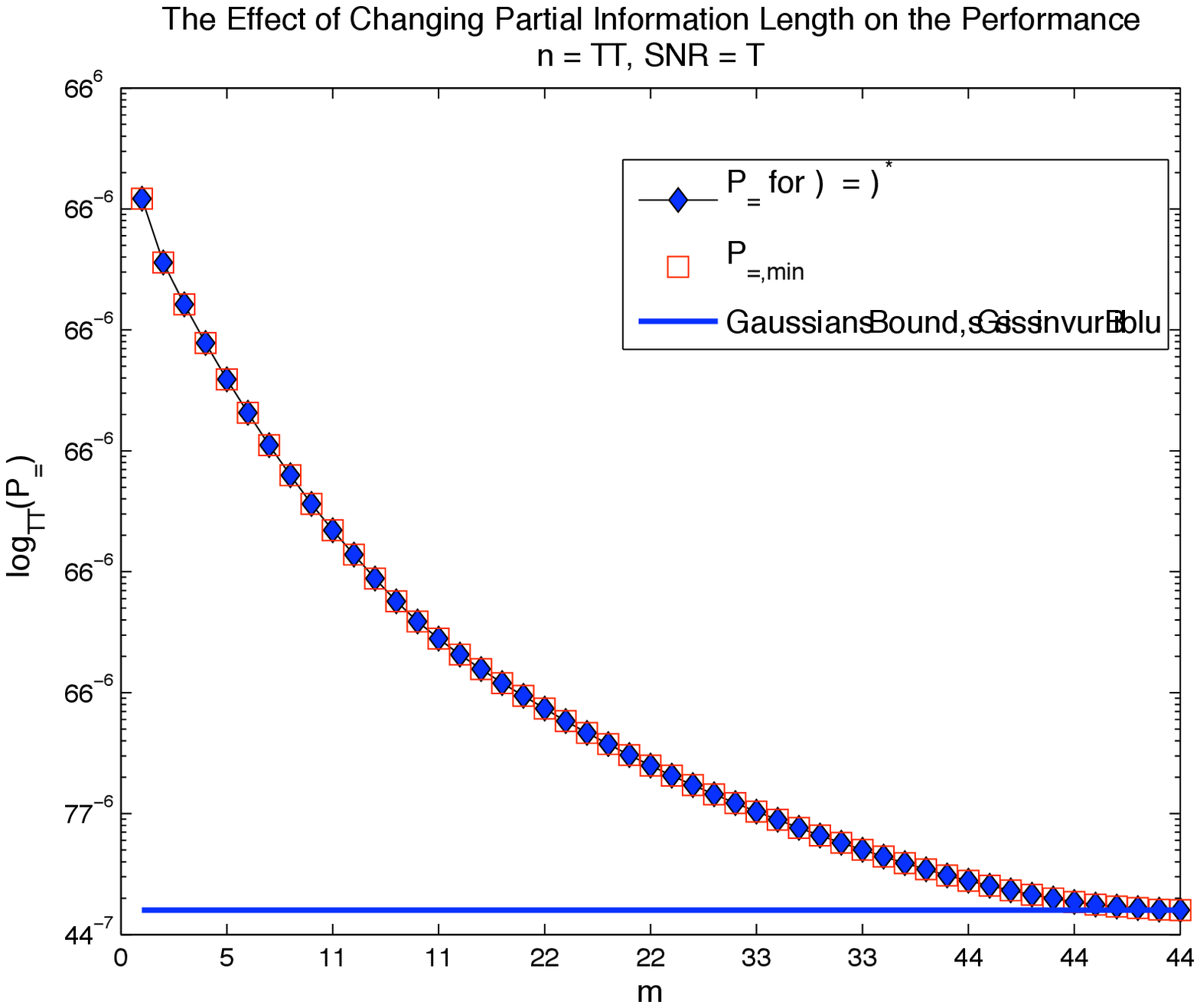}} \caption{Performance of $\bT_{opt}$ vs. $m$ (length of partial information)} 
\label{fig:PevsM}
\end{figure}

\noindent
\underline{\em $P_e$ vs. $n$}:
In this part we study the effect of changes in signal length on the performance of $\bT_{opt}$. The simulation
results, for various $\bS_x$ and $\bS_e$ all having uniformly distributed eigenvalues, are shown in Fig.~\ref{fig:PevsN}. 
At first glance, the results might seem counter-intuitive. 
The crucial point is that since $m$ (the dimension of the partial information space) is constant, 
as $n$ increases we get more degrees of freedom to construct  $\bT_{opt}$ (i.e. 
the number of eigenvalues of $\bP$ increases and so does (\ref{eq:optimal_det_val}), improving the detector performance).  
%the cardinality of the set of eigenvalues of $\bP$ increases and hence we get a possibly higher optimality given by~(\ref{eq:optimal_det_val}) that $\bM^*$ achieves, hence optimality of corresponding $\bT_{opt,e}$ is improved). One can conclude that, since we increase the redundancy in the signal representing binary information, the performance of decoding is improved, which is similar to an argument for error-correction coding.
\begin{figure}[!htb]
\centering \centerline{\epsfxsize = 3in
\epsfbox{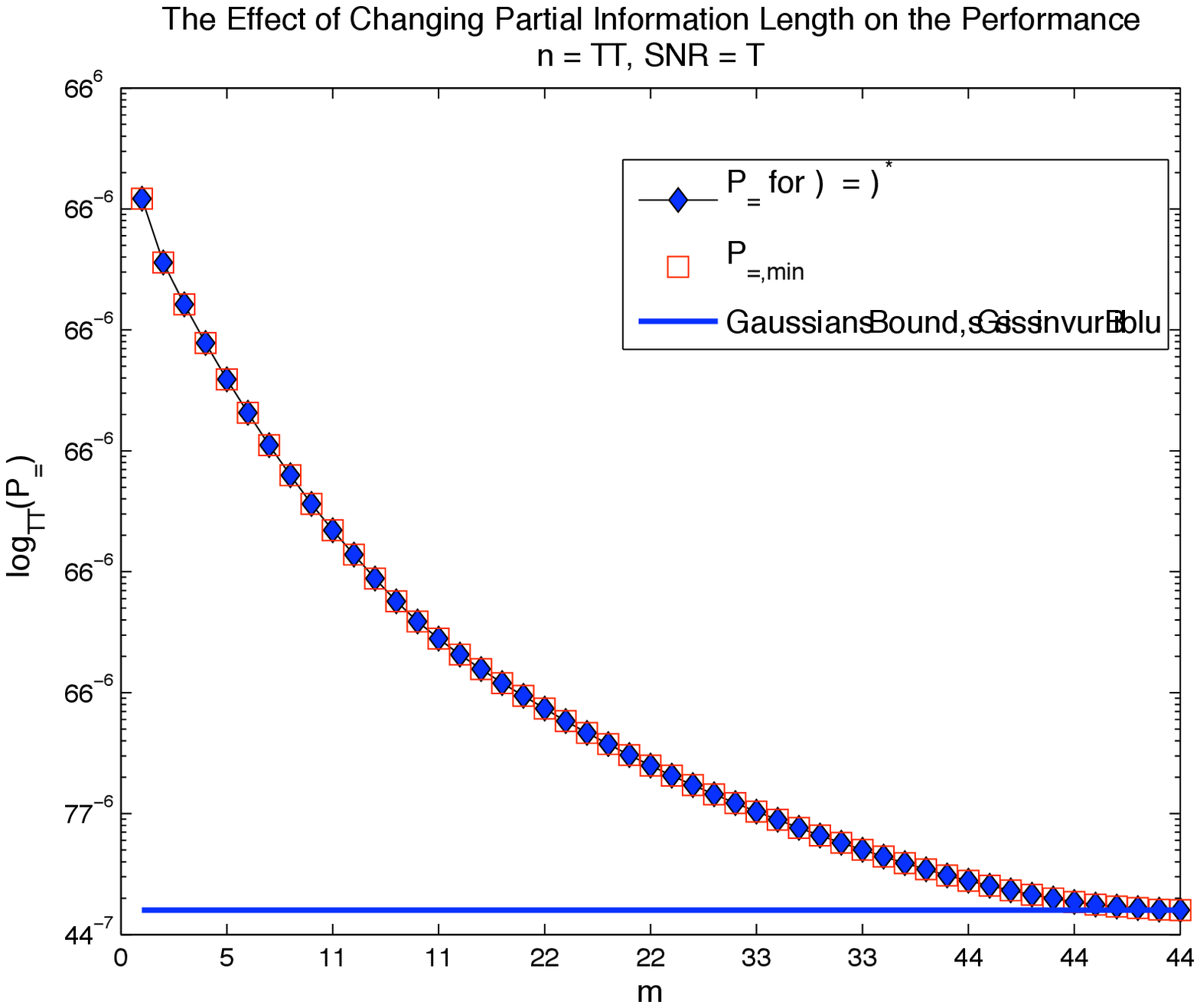}} \caption{Performance of $\bT_{opt}$ vs. $n$ (signal length)} 
\label{fig:PevsN}
\end{figure}

\Section{Conclusions} 
\label{sec:conclusions} 

We introduce the concept of communication with partial information. 
The main idea is that the codebooks used by the transmitter and the receiver are different. 
This concept  is different from that of communication with side information, where 
the utilized codebooks are the same but there is extra information available to one of the 
communicating parties. 

Within the context of communication with partial information, we particularly concentrate
on a binary detection theoretic scenario. The transmitter sends one of the two codewords
(which are independent realizations of a colored multivariate Gaussian distribution) 
to the additive colored Gaussian noise channel. 
The receiver acts as a detector, using 
{\em dimensionality reduced versions} of the encoder codewords, where the dimensionality 
reduction is achieved via a linear transform. 
We first find the optimal (in the sense of probability of error) 
detection rule. Then we derive the optimal class of linear transforms in the sense of 
the expected value of the Chernoff bound on the conditional probability of error of the detector. 

Although the  focus here is on binary detection, we believe that the proposed 
``communication with partial information'' covers several setups of interest, especially the cases
where there is an inherent asymmetry between the transmitter and the receiver due to the unbalanced 
limitations on the physical resources, such as memory and computational power. In our future research, 
we plan to explore various communication theoretic setups where asymmetry is a crucial feature.

\appendices
\renewcommand{\theequation}{\Roman{section}.\arabic{equation}}

%%%%%%%%%%%%%%%%%%%%%%%%%%%%%%%%%%%%%%%%%%%%%%%%%%%%%%%%

\section{Proof of Theorem~\ref{thm:cond_pro_err}}
\label{prf:cond_pro_err}

Throughout the proof, we use the definitions of $\by_i \eqdef \bx_i + \bbe$ for $i \in \left\{ 0 , 1 \right\}$. 
Accordingly, we use $\bm_{ y_i | z_i } = \textrm{E} \left( y_i | z_i \right)$ and $\bS_{y_i | z_i} = \textrm{Cov} \left( y_i | z_i \right)$. 
We start with the following lemma. 

\begin{lemma}
For $i \in \left\{ 0 , 1 \right\}$, conditioned on $\bz_i$, $\by_i$ is a normal random vector. Furthermore 
\be
%\bS_{y | z} = 
\bS_{y_i| z_i}  = \bS_x+\bS_e-\bS_x\bT^T(\bT\bS_x\bT^T)^{-1}\bT\bS_x, 
\label{eq:covariance-matrix}
\ee
is independent of $i$ and positive definite.
\label{lem:pos-def}
\end{lemma}

\begin{proof}
The crucial point is to show that, for $i \in \left\{ 0 , 1 \right\}$, $\by_i$ and $\bz_i$ are jointly normal with a positive definite covariance matrix. 
First, consider $\left[ \begin{array}{c} \bx_i \\ \bbe \end{array} \right] \in\bbR^{2n}$. Since $\bx_i$ and $\bbe$ are both normal and are independent, 
they are also jointly normal  with zero mean and the covariance matrix of 
$\bH \eqdef \left[\begin{array}{cc}\bS_x & \bzero \\ \bzero & \bS_e\end{array}\right] \in \bbR^{2n\times 2n}$. 
Note that, $\bH$ is clearly positive definite, since for any 
$\bv = \left[ \begin{array}{c} \bv_1 \\ \bv_2 \end{array} \right] \in \bbR^{2n}$
where $\bv_1 , \bv_2 \in \bbR^n$, 
$\bv^T\bH \bv = \bv_1^T\bS_x\bv_1+\bv_2^T \bS_e\bv_2\geq0$ by the positive definiteness of $\bS_x$ and $\bS_e$ (that we assumed). 
By the same token, $\left[ \bv_1^T \bS_x \bv_1 + \bv_2^T \bS_e \bv_2 = 0 \right] \iff  \left[ \bv_1 = \bv_2 = \bzero \right] \iff \left[ \bv = \bzero \right]$, yielding the positive definiteness of $\bH$.\\
Now, consider the linear transformation from the normal random vector $\left[ \begin{array}{c} \bx_i  \\ \bbe \end{array} \right] \in \bbR^{2n}$ 
to the vector $\left[ \begin{array}{c} \by_i \\ \bz_i \end{array} \right]\in \bbR^{n+m}$ represented by $\bF = \left[\begin{array}{cc}\bI_n & \bI_n \\\bT & \bzero_{m\times n}\end{array}\right] \in \bbR^{\left(n+m\right)\times 2n}$, where $\bzero_{m\times n}$ denotes the $m\times n$ zero matrix. 
This linear transform establishes the normality of $\left[ \begin{array}{c} \by_i \\ \bz_i \end{array} \right] \in \bbR^{n+m}$ (by the properties 
of jointly normal random vectors) with zero mean and the covariance matrix of $\bF\bH\bF^T = \left[ \begin{array} {cc} 
\bS_x + \bS_e & \bS_xT^T \\ \bT \bS_x & \bT \bS_x \bT^T \end{array} \right]$. To deduce the positive definiteness of this covariance matrix, i.e., 
$\bF\bH\bF^T$, it is sufficient to show that $\bF$ is full rank. 
This stems from the fact that if $\bF$ is full rank (i.e., if $r\left(\bF\right)=m+n$ since $m<n$), for any nonzero vector $\bs \in \bbR^ {m+n}$ we have $\bF^T\bs =  \bw \neq \bzero \in \bbR^{2n}$ since $\bF^T$ has a trivial {\em null-space}, so we end-up with $\bs^T\bF\bH\bF^T\bs = \bw^T\bH\bw > 0$ by the positive definiteness of $\bH$. 

To establish the full-rank property of $\bF$ (equivalent to having ``$\bF^T$ has a trivial null-space''), consider 
$\ba = \left[ \begin{array}{c} \ba_1 \\ \ba_2 \end{array} \right] \in\bbR^{m+n}$ where $\ba_1\in\bbR^n$ and $\ba_2\in\bbR^m$. 
In this case, $\bF^T\ba = \left[ \begin{array}{c}  \ba_1+\bT^T\ba_2 \\ \ba_1 \end{array} \right]$. 
Suppose there exists some $\ba \neq \bzero$ such that $\bF^T \ba = \bzero$. This implies, $\ba_1 = \bzero$ and $\bT^T \ba_2 = \bzero$. 
However, since $r(\bT)=m$, $\left[ \bT^T\ba_2=\bzero \right] \iff \left[ \ba_2=\bzero \right]$. 
Therefore, $\left[ \bF^T \ba = \bzero \right] \iff \left[ \ba = \bzero \right]$ and hence contradiction. Thus, $\bF$ is necessarily full-rank
implying positive-definiteness of the covariance matrix of $\left[ \begin{array}{c} \by_i \\ \bz_i \end{array} \right]$, i.e., $\bF \bH \bF^T$. 

Finally, the normality of $\left[\by_i\hspace{0.05in}|\hspace{0.05in}\bz_i\right]$ follows from the properties of 
normal distributed random variables. The positive definiteness of the corresponding covariance 
matrix $\bS_{y_i|z_i} = \bS_x+\bS_e-\bS_x\bT^T(\bT\bS_x\bT^T)^{-1}\bT\bS_x$ follows from the fact that it is the inverse of 
a principal submatrix of the inverse of $\bF\bH\bF^T$, which is positive definite (see (7.1.2) and (7.7.5) in~\cite{HornJohn99}).
Also, $\bS_{y_i|z_i}$ is clearly independent of $i \in \left\{ 0 , 1 \right\}$.  
\end{proof}

Per Lemma~\ref{lem:pos-def}, since $\bS_{y|z} = \bS_{y_i|z_i}$ is positive definite, it is invertible and it has an invertible square root. 

\begin{remark}
From properties of normal random vectors, we have  
\be
\bm_{y_i|z_i}=\textrm{E} \left( \by_i \, | \, \bz_i \right) =  \bS_x\bT^T(\bT\bS_x\bT^T)^{-1}\bz_i.
\label{eq:conditional-mean}
\ee
\label{rem:cond-mean}
\end{remark}

Now let $\beta \eqdef [(2
\pi)^{n/2}\det(\bS_{y|z})^{1/2}]^{-1}$, 
$\theta \eqdef \left( \bm_{y_0|z_0} - \right.$ $\left. \bm_{y_1|z_1} \right)^T \bS_{y|z}^{-1} \by$, and $\alpha_i \eqdef (\by-\bm_{y_i|z_i})^T \bS_{y|z}^{-1}
(\by-\bm_{y_i|z_i})$, $\kappa_i \eqdef \bm_{y_i|z_i}^T \bS_{y|z}^{-1} \bm_{y_i|z_i}$ for $i = 0, 1$; where $\bS_{y|z}$ and $\bm_{y_i|z_i}$ are given in~(\ref{eq:covariance-matrix}) and~(\ref{eq:conditional-mean}), respectively. Then, using Lemma~\ref{lem:pos-def} and Remark~\ref{rem:cond-mean}, given $\by$ is observed we have 
 \begin{equation}
p\left(\by_i | \bz_i \right)\Big|_{y_i=y} = 
%\frac{1}{(2 \pi)^{n/2}\det(\bS_{y|z})^{1/2}} 
\beta 
\exp \left[- \frac{\alpha_i}{2}\right].
\label{eq:cond-dist-y}
\end{equation} 
Then, using the above distribution of $\left[\by_i \, | \, \bz_i \right]$ the maximum likelihood detection rule~(\ref{eq:MLdetection-rule2}) can be written as
\be
\beta \exp \left[-\frac{\alpha_0}{2}\right]
\mathop{\gtrless}_{H_1}^{H_0}
\beta \exp \left[-\frac{\alpha_1}{2}\right] \nonumber
\ee
which is equivalent to (\ref{eq:ML_detect_rule}) since $\det\left(\bS_{y|z}\right) \neq 0$ and $\exp(.)$ is a 
strictly increasing function in its argument. Moreover, 
%Next,  we proceed with the calculation of the error probability conditioned on $H_0$: 
\begin{eqnarray}
%\Pr \left[ \mbox{error} \, \big| \, H_0 \right]
P_{e | H_0 }
& = & 
\Pr  \left[ \alpha_0
> \alpha_1 \; 
\Big| \; \by \sim \cN \left( \bm_{y_0|z_0} , \bS_{y|z} \right) \right]
%\label{eq:ML-rule_reduced} \\
\nonumber \\
& = & 
\Pr \left[ \theta < \frac{\kappa_0 - \kappa_1}{2} \; 
\Big| \; \by \sim \cN(\bm_{y_0|z_0}, \bS_{y|z}) \right] , 
\nonumber
\end{eqnarray}
where $P_{e|H_0}$ denotes the probability of error conditioned on $H_0$. 
Here, conditioned on $H_0$, 
the random variable 
%$\theta \eqdef \left( \bm_{y_0|z_0} - \bm_{y_1|z_1} \right)^T \bS_{y|z}^{-1} \by \in \bbR$ 
$\theta$
is normal 
since  $\left( \bm_{y_0|z_0} - \bm_{y_1|z_1} \right)^T \bS_{y|z}^{-1}$ is a linear transformation from $\bbR^n$ to $\bbR$
and $\by | H_0$ is normal. 
Conditioned on $H_0$, the mean and variance of $\theta$ are given by 
\begin{eqnarray}
\mu_{\theta | H_0} & = &  
%\textrm{E} \left( \left( \bm_{y_0|z_0} - \bm_{y_1|z_1} \right)^T \bS_{y|z}^{-1} \by \;  \Big| \; \by \sim \cN(\bm_{y_0|z_0}, \bS_{y|z}) \right) = 
\left( \bm_{y_0|z_0} - \bm_{y_1|z_1} \right)^T \bS_{y|z}^{-1} \bm_{y_0|z_0} ,
\nonumber \\
\sigma_{\theta | H_0}^2  &  = & 
%\left( \bm_{y_0|z_0} - \bm_{y_1|z_1} \right)^T \bS_{y|z}^{-1} \bS_{y|z} \bS_{y|z}^{-1} \left( \bm_{y_0|z_0} - \bm_{y_1|z_1} \right) = 
\left( \bm_{y_0|z_0} - \bm_{y_1|z_1} \right)^T \bS_{y|z}^{-1} \left( \bm_{y_0|z_0} - \bm_{y_1|z_1} \right) . 
\nonumber
\end{eqnarray}
%\begin{eqnarray}
%\mu_{\theta | H_0} & = &  
%\textrm{E} \left( \left( \bm_{y_0|z_0} - \bm_{y_1|z_1} \right)^T \bS_{y|z}^{-1} \by \;  \Big| \; \by \sim \cN(\bm_{y_0|z_0}, \bS_{y|z}) \right) = 
%\left( \bm_{y_0|z_0} - \bm_{y_1|z_1}  \right)^T \bS_{y|z}^{-1} \textrm{E} \left( \by \;  \Big| \; \by \sim \cN(\bm_{y_0|z_0}, \bS_{y|z}) \right) 
%\nonumber \\
%& = & \left( \bm_{y_0|z_0} - \bm_{y_1|z_1} \right)^T \bS_{y|z}^{-1} \bm_{y_0|z_0} , 
%\nonumber
%\end{eqnarray}
%\[
%\sigma_{\theta | H_0}^2 = \left( \bm_{y_0|z_0} - \bm_{y_1|z_1} \right)^T \bS_{y|z}^{-1} \bS_{y|z} \bS_{y|z}^{-1} \left( \bm_{y_0|z_0} - \bm_{y_1|z_1} \right) 
%= \left( \bm_{y_0|z_0} - \bm_{y_1|z_1} \right)^T \bS_{y|z}^{-1} \left( \bm_{y_0|z_0} - \bm_{y_1|z_1} \right), 
%\]
%respectively. 
Then, $\Pr \left[ \mbox{error} \, | \, H_0 \right]$ is given in terms of the standard $Q$-function. 
As a result, after some algebraic manipulations we get 
\[
P_{e | H_0} =  \trQ \left( \frac{ \| \bS_{y|z}^{-1/2} \left( \bm_{y_0|z_0} -\bm_{y_1|z_1} \right) \| }{2} \right) 
= \trQ \left( \frac{ \sigma_{\theta | H_0} }{2} \right)
. 
\]
%\bean
%\Pr \left[ \mbox{error} \, | \, H_0 \right] 
%& = & \trQ \left( \frac{ \left( \bm_{y_0|z_0} - \bm_{y_1|z_1} \right)^T \bS_{y|z}^{-1} \bm_{y_0|z_0} - \left( \bm_{y_0|z_0}^T \bS_{y|z}^{-1} 
%\bm_{y_0|z_0} - \bm_{y_1|z_1}^T \bS_{y|z}^{-1} \bm_{y_1|z_1} \right)/2}{ \left[ \left( \bm_{y_0|z_0} - \bm_{y_1|z_1} \right)^T \bS_{y|z}^{-1} 
%\left( \bm_{y_0|z_0} - \bm_{y_1|z_1} \right) \right]^{\frac{1}{2}}} \right) , \\
%& = & \trQ \left( \frac{ \left[ \left( \bm_{y_0|z_0} - \bm_{y_1|z_1} \right)^T \bS_{y|z}^{-1} \left( \bm_{y_0|z_0} -\bm_{y_1|z_1} \right) \right]^{\frac{1}{2}}}{2} \right) , \\
%& = & \trQ \left( \frac{ \| \bS_{y|z}^{-1/2} \left( \bm_{y_0|z_0} -\bm_{y_1|z_1} \right) \| }{2} \right)  \; .
%\eean
Furthermore, from symmetry, we have 
$P_{e | \bz_0 , \bz_1 } = P_{e | H_0}$.
%$\Pr \left[ \mbox{error}  \, | \, \left\{ \bz_0 , \bz_1 \right\} \right] 
%= \Pr \left[ \mbox{error} \, | \, H_0 \right]$. Hence the proof.  
\qed

%%%%%%%%%%%%%%%%%%%%%%%%%%%%%%%%%%%%%%%%%%%%%%%%%%%%%%%%

\section{Proof of Proposition~\ref{prop:Exp_Chernoff_bound}}
\label{prf:Exp_Chernoff_bound}

First, we recall the standard Chernoff bound on $\trQ \left( \cdot \right)$  function: 
$\trQ \left( x \right) \leq \frac{1}{2} \exp \left(-\frac{x^2}{2}\right)$ for $x \geq 0$ \cite{Loeve}. 
Then, (\ref{eq:chernoff_bound}) is obvious via using it in (\ref{eq:cond_prob_err}). 
Next, we have 
\begin{eqnarray}
\hspace{-0.5cm} P_e & \leq & \textrm{E}_{\{\bz_0,\bz_1\}} 
\left[ \frac{1}{2} \exp  \left(-\frac{ 
%\| \bS_{ y | z }^{-1/2} \left( \bm_{ y_0 | z_0 }  - \bm_{ y_1 | z_1 } \right) \|^2 
\sigma^2_{\theta | H_0}
}{8}\right) \right], 
\label{eq:Exp_Chernoff_bound_temp1} \\
& = & 
\textrm{E}_{\{\bz_0,\bz_1\}}  \left[ 
\frac{1}{2} \exp \left( -\frac{\left[ \bA \bgam\right]^T\bB^{-1}\left[ \bA \bgam\right]}{8} \right)
\right]  , 
\label{eq:Exp_Chernoff_bound_temp2} \\
& = & 
\textrm{E}_{\bgam}  \left[ 
\frac{1}{2} \exp \left( -\frac{\left[ \bA \bgam\right]^T\bB^{-1}\left[ \bA \bgam\right]}{8} \right)
\right] ,
\label{eq:Exp_Chernoff_bound_temp3} \\
& = & 
\int_{\bbR^{m}} \frac {1/2}{\left( 2\pi \right)^{\frac{m}{2}} \det \left(2\bS_z \right)^{\frac{1}{2}}}
\nonumber \\
& & 
\hspace{-0.55cm}
\exp \left(-\frac{1}{2}  \bgam^T \left[ 
%\frac{ (\bS_z)^{-1} }{2} 
\left( 2\bS_z \right)^{-1} 
+ \frac{ \bA^T \bB^{-1} \bA}{4} \right]\bgam  \right)
\mbox{d$\bgam$} ,
\label{eq:Exp_Chernoff_bound_temp4} 
\end{eqnarray}
where 
(\ref{eq:Exp_Chernoff_bound_temp1}) follows from using (\ref{eq:chernoff_bound}) in (\ref{eq:expected_error}), 
(\ref{eq:Exp_Chernoff_bound_temp2}) follows from using the definitions of $\bS_{ y | z }$, $\bm_{ y_0 | z_0 }$, 
$\bm_{ y_0 | z_0 }$ (cf. Theorem~\ref{thm:cond_pro_err}) and defining 
$\bA \eqdef  \bS_x\bT^T(\bT\bS_x\bT^T)^{-1}$,  $\bB \eqdef \bS_{y \, | \, z}$, $\bgam \eqdef \bz_0 - \bz_1$, 
(\ref{eq:Exp_Chernoff_bound_temp3}) follows since the only source of randomness is due to $\bgam$ per 
our reparametrization, 
(\ref{eq:Exp_Chernoff_bound_temp4}) follows since $\bgam \sim \cN \left(   \bzero, 2\bS_{z} \right)$
where  $\bS_z = \textrm{Cov} \left( \bz_0 \right) = \textrm{Cov} \left( \bz_1 \right) =
\bT \bS_x \bT^T$. 

Next, we proceed by showing the positive definiteness of the matrix 
$\left[ \left( 2\bS_z \right)^{-1} + \frac{ \bA^T \bB^{-1} \bA}{4} \right]^{-1}$, which would ensure that it is a valid covariance matrix. 
First, by  assumption, $\bS_x$ is positive definite and $\bT$ is full-rank. 
Hence, using similar steps to the ones that are used in the proof of positive definiteness of $\bF \bH \bF^T$ within the proof of 
Lemma~\ref{lem:pos-def}, we conclude that $\bS_z = \bT \bS_x \bT^T$ is positive definite. 
Furthermore, 
%$\left[ \mbox{positive definiteness of $\bS_z = \bT \bS_x \bT^T$} \right] \iff \left[ 
%\mbox{positive} \right.$ $\left. \mbox{ definiteness of $2\bS_z$}\right]
%\iff \left[ \mbox{the positive definiteness of $\left( 2\bS_z \right)^{-1}$} \right]$ 
$\left[ \bS_z = \bT \bS_x \bT^T > 0  \right] \iff 
\left[  2\bS_z > 0 \right] \iff
 \left[ \left( 2\bS_z \right)^{-1} > 0  \right]$.
Next, note that $\bA$ is full-rank using straightforward linear algebra. 
Using this result and the positive definiteness of $\bB$, and applying similar arguments to those above, we
conclude that $\bA^T\bB^{-1}\bA$ is positive definite as well. 
Thus, $\left[ \left( 2\bS_z \right)^{-1} + \frac{ \bA^T \bB^{-1} \bA}{4} \right]^{-1}$ is positive definite since it is the inverse of the sum of two positive definite matrices, which is itself positive definite. 
As a result, the quantity $\left[ \left( 2\bS_z \right)^{-1} + \frac{ \bA^T \bB^{-1} \bA}{4} \right]^{-1}$ is a valid covariance matrix
and the integral (\ref{eq:Exp_Chernoff_bound_temp4}) converges, yielding
\[
P_e  \leq
%\frac{1}{2}\left(\frac{\det \left( \left[ \left( 2\bS_z \right)^{-1} + \frac{ \bA^T \bB^{-1} \bA}{4} \right]^{-1} \right)}
%{\det \left( 2\bS_z \right)} \right) ^{\frac{1}{2}}
%= 
\frac{1}{2} \left\{ \det \left( \bI_m + \frac{\bS_z \bA^T \bB^{-1} \bA}{2}\right) \right\} ^{-\frac{1}{2}},
\]
by properties of determinants; hence the proof. 
\qed

%%%%%%%%%%%%%%%%%%%%%%%%%%%%%%%%%%%%%%%%%%%%%%%%%%%%%%%
%%%%%%%%%%%%%%%%%%%%%%%%%%%%%%%%%%%%%%%%%%%%%%%%%%%%%%%%

\section{Proof of Proposition~\ref{prop:optimal-T-prb-reduction}}
\label{prf:optimal-T-prb-reduction}

Our first goal is to show that $\bT^\ast \in \cS_\bT$. First, note that, $\bT^\ast$ is a $m \times n$ matrix by construction. 
Next, observe that, by definition $\bE \bD$ is a $m \times m$, non-singular matrix and $\bU_p^T \bL^{-1} \bF^T$ is a 
$n \times n$, non-singular matrix. Furthermore, $\bM^\ast$ is of size $n \times m$ and $r \left( \bM^\ast \right) = m$, i.e., 
it is full-rank by definition. Hence, $\bT^\ast$ is also of rank-$m$, implying that $\bT^\ast \in \cS_\bT$. 
Next, using  $\bS_x = \bF \bL^2 \bF^T$ and the definition of $\bT^\ast$, 
after some algebraic manipulations we get 
\begin{eqnarray}
\bF \bL \bU_p \bM^\ast 
&  = &  
\bS_x \left( \bT^\ast \right)^T \bE \bD^{-1} , 
\label{eq:optimal-T-prb-reduction-temp1}  \\
\bE \bD^{-2} \bE^T 
& = &   
\left( \bT^\ast \bS_x \left( \bT^\ast \right)^T \right)^{-1} . 
\label{eq:optimal-T-prb-reduction-temp2}
\end{eqnarray}
%\begin{equation}
%\bF \bL \bU_p \bM^\ast 
% = 
%\bS_x \left( \bT^\ast \right)^T \bE \bD^{-1} , 
%\label{eq:optimal-T-prb-reduction-temp1} 
%\end{equation}
%we have 
%\begin{equation}
%\bF \bL \bU_p \bM^\ast 
% = 
%\bF \bL \cdot \bL \bF^T \cdot \bF \bL^{-1} \bU_[ \bM^\ast 
%= 
%\bS_x \bF \bL^{-1} \bU_p \bM^\ast 
%= 
%\bS_x \bF \bL^{-1} \bU_p \bM^\ast \cdot \bD \bE^T \cdot \bE \bD^{-1} 
%= 
%\bS_x \left( \bT^\ast \right)^T \bE \bD^{-1} 
%\label{eq:optimal-T-prb-reduction-temp1} 
%\end{equation}
%\begin{equation}
%\bE \bD^{-2} \bE^T 
%=   
%\left( \bE \bD^2 \bE^T \right)^{-1} 
%= 
%\left( \bE \bD \cdot \left( \bM^\ast \right)^T \bU_p^T \cdot \bL^{-1} \bF^T \cdot \bF \bL^2 \bF^T \cdot \bF \bL^{-1} \cdot 
% \bU_p  \bM^\ast \bD \bE^T \right)^{-1} 
%= 
%\left( \bT^\ast \bS_x \left( \bT^\ast \right)^T \right)^{-1}
%\label{eq:optimal-T-prb-reduction-temp2}
%\end{equation}
\begin{lemma} 
For any  $\bM \in \cS_\bM$, letting $\bT = \bE \bD \bM^T \bU_p^T \bL^{-1} \bF^T$, where 
$\bE \in \bbR^{m\times m}$ is an arbitrary  unitary matrix and 
$\bD \in \bbR^{m \times m}$ is an arbitrary diagonal positive-definite matrix, 
we have $G \left( \bM \right) = J \left( \bT \right)$. 
\label{lem:optimal-T-prb-reduction-templemma}
\end{lemma}
\begin{proof}
We have 
\begin{eqnarray}
G \left( \bM \right) 
& = & 
%\left( \frac{1}{2} \right)^m \prod_{i=1}^m \left[ 1 + \frac{1}{\lambda_i \left(  \bM^T \hat{\bL}_p \bM \right) } \right]
%= 
%\left( \frac{1}{2} \right)^m \prod_{i=1}^m \lambda_i \left( \bI_m + \left(  \bM^T \hat{\bL}_p \bM \right)^{-1} \right) , 
%\label{eq:optimal-T-prb-reduction-temp3} \\ 
\left( \frac{1}{2} \right)^m \det \left[ \bI_m + \left(  \bM^T \hat{\bL}_p \bM \right)^{-1} \right] ,
%= 
%\left( \frac{1}{2} \right)^m \det \left[ \bI_m + \left(  \bM^T \bM - \bM^T \bL_p^{-1} \bM \right)^{-1} \right] , 
\label{eq:by_definition_of_determinant}\\
%\left( \frac{1}{2} \right)^m \det \left[ \bI_m + \left( \bI_m - \bM^T \bU_p^T \bU_p \bL_p^{-1} \bU_p^T \bU_p \bM \right)^{-1} \right] , 
%\label{eq:optimal-T-prb-reduction-temp4} \\
%& = & 
%2^{-m} \det \left[ 2 \bI_m + \left( \bU_p \bM \right)^T \left( \bP - \bU_p \bM  \left( \bU_p \bM \right)^T \right)^{-1} 
%\bU_p \bM \right]
%\label{eq:optimal-T-prb-reduction-temp5} \\
%& = & 
%\det \left[ \bI_m + \frac{1}{2} \bE \bD \bM^T \bU_p^T \left( \bL^{-1} \bF^T \left( \bS_x + \bS_e \right) \bF \bL^{-1} 
%- \bU_p \bM \bM^T \bU_p^T \right)^{-1} \bU_p \bM \bD^{-1} \bE^T \right]
%\label{eq:optimal-T-prb-reduction-temp6} \\
& = & 
\det \bigg[ \bI_m + \frac{1}{2} \bE \bD \bM^T \bU_p^T \bL^{-1} \bF^T \bF \bL \bL \bF^T
\nonumber\\
&& 
\left[ \bS_x + \bS_e - \bS_x \bT^T \bE \bD^{-2} \bE^T \bT \bS_x \right]^{-1}
\nonumber \\
&&
\bS_x \bT^T \bE \bD^{-2} \bE^T\bigg ]
\label{eq:optimal-T-prb-reduction-temp7} \\
%& = & 
%\det \left[ \bI_m + \frac{1}{2} \bT \bS_x \left[ \bS_x + \bS_e - \bS_x \bT^T \left( \bT \bS_x \bT^T \right)^{-1} \bT \bS_x \right]^{-1}
%\bS_x \bT^T \left( \bT \bS_x \bT^T \right)^{-1} \right]
%\label{eq:optimal-T-prb-reduction-temp8} \\
& = & 
J \left( \bT \right)  
\label{eq:substitute_and_calculate_J}
\end{eqnarray}

%where 
%(\ref{eq:optimal-T-prb-reduction-temp3}) follows since $\bM$ is full-rank, $\hat{\bL}_p$ is positive definite, and 
%thus $ \bM^T \hat{\bL}_p \bM$ is positive definite symmetric, 
%(\ref{eq:optimal-T-prb-reduction-temp4}) follows since $\bM$ is orthonormal and $\bU_p$ is unitary, 
%(\ref{eq:optimal-T-prb-reduction-temp5}) follows from using the SVD of $\bP$ and the matrix inversion lemma, 
%(\ref{eq:optimal-T-prb-reduction-temp6}) follows from using the properties of determinants and the definition of $\bP$,  
%(\ref{eq:optimal-T-prb-reduction-temp7}) follows from using (\ref{eq:optimal-T-prb-reduction-temp1}), 
%(\ref{eq:optimal-T-prb-reduction-temp8}) follows from the definition of $\bT$, the SVD of $\bS_x$ and using 
%(\ref{eq:optimal-T-prb-reduction-temp2}).  
where 
(\ref{eq:by_definition_of_determinant}) follows from the definition of determinant and properties of positive definite matrices; 
(\ref{eq:optimal-T-prb-reduction-temp7}) follows from our auxiliary definitions, properties of the defined matrices and the matrix inversion lemma;
(\ref{eq:substitute_and_calculate_J}) follows from the substitution of the auxiliary matrices in (\ref{eq:optimal-T-prb-reduction-temp7}). 
\end{proof}

\begin{lemma} 
For any  $\bT \in \cS_\bT$, there exists $\bM \in \cS_\bM$, such that $J \left( \bT \right) = G \left( \bM \right)$. 
\label{lem:optimal-T-prb-reduction-templemma2}
\end{lemma}
\begin{proof}
For any $\bT \in \cS_\bT$, let $\tilde{\bE}$ and $\tilde{\bD}$ be given by the SVD of $\bT \bS_x \bT^T$, i.e., 
$\bT \bS_x \bT^T = \tilde{\bE} \tilde{\bD}^2 \tilde{\bE}^T$. Naturally, $\tilde{\bE} \in \bbR^{m \times m}$
and $\tilde{\bD} \in \bbR^{m \times m}$ are unitary and positive-definite diagonal, respectively. Then let 
$\bM \eqdef \bU_p^T \bL \bF^T \bT^T \tilde{\bE} \tilde{\bD}^{-1}$. \\
First, we show that $\bM \in \cS_\bM$. Clearly, $\bM \in \bbR^{n \times m}$. Here, 
\begin{eqnarray}
\bM^T \bM 
& = &  
\bD^{-1} \tilde{\bE}^T \bT \bF \bL \bU_p \bU_p ^T \bL \bF^T \bT^T \tilde{\bE} \tilde{\bD}^{-1}
\nonumber \\
& = &
\bD^{-1} \tilde{\bE}^T \bT \bF \bL^2 \bF^T \bT^T \tilde{\bE} \tilde{\bD}^{-1}
\nonumber \\
 & = & 
\bD^{-1} \tilde{\bE}^T \bT \bS_x \bT^T \tilde{\bE} \tilde{\bD}^{-1} , 
\nonumber \\
& = & 
\bD^{-1} \tilde{\bE}^T \tilde{\bE} \tilde{\bD}^2 \tilde{\bE}^T \tilde{\bE} \tilde{\bD}^{-1}
= 
\bI_m , 
\nonumber 
\end{eqnarray}
implying $\bM \in \cS_\bM$. 
Now, note that we have $\bT = \tilde{\bE} \tilde{\bD} \bM^T \bU_p^T \bL^{-1} \bF^T$ due to the way $\bM$ was defined; 
this means $\bT$ is of the functional form given in the statement of Prop.~\ref{prop:optimal-T-prb-reduction}. 
Also,  $\tilde{\bE}$ is unitary and $\tilde{\bD}$ is diagonal and positive definite. 
Therefore, we necessarily have $G \left( \bM \right) = J \left( \bT \right)$ per 
Lemma~\ref{lem:optimal-T-prb-reduction-templemma}. Hence the proof. 
\end{proof}
Now, we go back to the proof of Prop.~\ref{prop:optimal-T-prb-reduction} and use proof by contradiction. 
Suppose, there exists some $\bar{\bT} \in \cS_\bT$ such that $J \left( \bar{\bT} \right) > J \left( \bT^\ast \right)$. 
By Lemma~\ref{lem:optimal-T-prb-reduction-templemma}, we necessarily have $J \left( \bT^\ast \right) = 
G \left( \bM^\ast \right)$. Furthermore, by Lemma~\ref{lem:optimal-T-prb-reduction-templemma2}, 
there exists $\bar{\bM} \in \cS_M$ such that $J \left( \bar{\bT} \right) = G \left( \bar{\bM} \right)$. But this implies 
$ G \left( \bar{\bM} \right) = J \left( \bar{\bT} \right)  > J \left( \bT^\ast \right) = G \left( \bM^\ast \right)$ which contradicts
with the way $\bM^\ast$ was defined in the first place. Hence contradiction and proof. 
 \qed

%%%%%%%%%%%%%%%%%%%%%%%%%%%%%%%%%%%%%%%%%%%%%%

\section{Proof of Proposition~\ref{prop:opt-reduced-soln}}
\label{prf:opt-reduced-soln}

First, observe that, $G \left( \bM \right)$ is a product of positive real numbers since $\lambda_i \left( \bM^T \hat{\bL}_p \bM \right)
> 0$ for all $i$ by the positive definiteness of $\bM^T \hat{\bL}_p \bM$ (because $\bM$ is orthonormal, full-rank and $\hat{\bL}_p$
is positive definite). So, in order to maximize $G \left( \bM \right)$, 
we follow the strategy of maximizing each positive factor $\left( 1 + \frac{1}{\lambda_i \left( \bM^T 
\hat{\bL}_p \bM \right) } \right)$ for all $i$, which clearly is equivalent to minimizing $\lambda_i \left( \bM^T \hat{\bL}_p 
\bM \right)$ for all $i$. Here, let $\bQ \in \bbR^{m \times m}$ denote a permutation matrix such that 
$\hat{\bL}_p = \bQ^T \hat{\hat{\bL}}_p \bQ$, the matrix $\hat{\hat{\bL}}_p$ is diagonal, 
and its eigenvalues (i.e., the diagonal entries)  are in non-decreasing order
\footnote{See \cite{HornJohn99} for the existence of such a $\bQ$. Note that, such a $\bQ$ is unique iff the eigenvalues 
of $\hat{\bL}_p$ are distinct.}. Then, $G \left( \bM \right)$ can be rewritten as 
\begin{equation}
G \left( \bM \right) \eqdef 2^{-m} \prod_{i=1}^m \left[ 1 + \frac{1}{\lambda_i 
\left( \bM^T \bQ^T \hat{\hat{\bL}}_p \bQ \bM \right) }
\right] .
\label{eq:opt-reduced-soln-temp1}
\end{equation}
Next, we recall the {\em Poincar\'e seperation theorem} (see \cite{HornJohn99}, pp. 190--191) which is crucial in completing the proof. 
\begin{theorem}
Let $\bA \in \bbR^{n\times n}$ be symmetric, and let $m$ be a given integer with $1\leq m \leq n$, and $\bB_m = \bU^T \bA \bU$, where $\bU \in \bbR^{n \times m}$ is orthonormal. If eigenvalues of $\bA$ and $\bB_m$ are arranged in non-decreasing order, we have
\be
\lambda_i \left( \bA \right) \leq \lambda_i \left( \bB_m \right) \leq \lambda_{i+n-m} \left( \bA \right) \hspace{0.08in} i = 1,2,...,m
\label{eq:Poincare}
\ee
\end{theorem}
\noindent \\
Using (\ref{eq:Poincare}) in (\ref{eq:opt-reduced-soln-temp1}), we get 
\begin{equation}
G \left( \bM \right) \leq 2^{-m} \prod_{i=1}^m \left[ 1 + \frac{1}{\lambda_i \left( \hat{\hat{\bL}}_p \right) } \right] . 
\label{eq:opt-reduced-soln-temp2}
\end{equation}
Choosing $\bQ \bM = \left[ \bI_m \; \bzero_{ m \times \left( n - m \right) } \right]^T$ clearly satisfies 
$\lambda_i  \left( \bM^T \bQ^T \hat{\hat{\bL}}_p \bQ \bM \right) = \lambda_i \left( \hat{\hat{\bL}}_p \right)$ for $1 \leq i \leq m$, 
thereby achieving (\ref{eq:opt-reduced-soln-temp2}) with equality. Furthermore, since eigenvalues are invariant under 
similarity transformations, for any unitary $\bGam \in \bbR^{m \times m}$
choosing $\bQ \bM = \left[ \bGam_{m \times m}^T  \; \bzero_{ m \times \left( n - m \right) } \right]^T$
also satisfies (\ref{eq:opt-reduced-soln-temp2}) with equality.
Also, the resulting \\
$\bM = \bQ^T  \left[ \bGam_{m \times m}^T  \; \bzero_{ m \times \left( n - m \right) } \right]^T$ clearly 
satisfies $\bM \in \cS_\bM$. Hence, any such $\bM$ is a solution to (\ref{eq:optimal-reduct-M}) where the maximum value is the %right hand side 
RHS of  (\ref{eq:opt-reduced-soln-temp2}). 
\qed

\section*{Acknowledgement}
Authors wish to thank Tamer Ba\c{s}ar, Vishal Monga, Sviatoslav Voloshynovskiy, Oleksiy Koval, Serdar Kozat 
and Serdar Y\"{u}ksel for various helpful discussions and comments.

\end{document}